\documentclass{article}
\usepackage{amsmath}
\usepackage{amssymb}
\usepackage{amsthm}
\usepackage{authblk} 
\usepackage{lineno}
\newtheorem{theorem}{Theorem}
\newtheorem{lemma}[theorem]{Lemma}

\title{Efficient Deterministic Single Round Document Exchange for Edit Distance}
\author{Djamal Belazzougui}
\affil{DTISI, CERIST Research center, 05 rue des trois freres Aissou, Benaknoun, Algiers, Algeria} 

\begin{document}
\maketitle
\begin{abstract}
Suppose that we have two parties that possess each 
a binary string. Suppose that the length of the first string 
(document) is $n$ and that the two strings (documents) 
have edit distance (minimal number 
of deletes, inserts and substitutions needed 
to transform one string into the other) at most 
$k$. The problem we want to solve is to devise 
an efficient protocol in which the first party sends a single message
that allows the second party to guess the first party's string. 
In this paper we show an efficient deterministic protocol for 
this problem. 
The protocol runs in time $O(n\cdot \mathtt{polylog}(n))$ and 
has message size $O(k^2+k\log^2n)$ bits. 
To the best of our knowledge, ours is the first 
\emph{efficient} deterministic protocol for this problem, 
if efficiency is measured in both the message 
size and the running time. As an immediate application of 
our new protocol, we show a new error correcting code that is efficient 
even for large numbers of (adversarial) edit errors. 
\end{abstract}
\section{Introduction}
Suppose that we have two parties that possess each 
a binary string. Suppose that the length of the first string 
(document) is $n$ and that the two strings (documents) 
have edit distance (minimal number 
of deletes, inserts and substitutions needed 
to transform one string into the other) at most 
$k$. The problem we want to solve is to devise 
an efficient protocol in which the first party sends a single message
that allows the second party to guess the first party's string. 
We call this problem the \emph{one-way document exchange 
under the edit distance}. 
In this paper, we answer an open question raised in~\cite{brakensiek2015efficient}  
by showing a deterministic solution to this problem with message 
size $O(k^2+k\log^2n)$ bits and encoding-decoding 
time $O(n\cdot \mathtt{polylog}(n))$~\footnote{
The decoding time is actually $O((n+m)\cdot \mathtt{polylog}(n+m))$, 
where $m$ is message size. However, we can safely assume 
that $m\leq n$, if the message size of the protocol exceeds $n$, 
then we can just send the original string.}\footnote{$f(n)=\mathtt{polylog}(n)$
if and only if $f(n)=\log^c(n)$ for some constant $c$.}. 
This result is to be compared to previous randomized schemes 
that achieve $O(k^2\log n)$~\cite{chakraborty2015low}, $O(k\log n\log (n/k))$~\cite{IMS05},  
and $O(k\log^2n\log^*n)$ bits~\cite{jowhari2012efficient}. 
We note that an optimal code should use $\Theta(k\log (n/k))$ bits~\cite{orlitsky1993interactive,brakensiek2015efficient}, 
and in fact this can be achieved with a protocol that runs in time exponential in $n$.
However, to the best of our knowledge no such deterministic code with polynomial time 
decoding and encoding is known for arbitrary values of $k$. 
We are not aware of any deterministic protocol with message size polynomial 
in $k\log n$ and encoding-decoding time polynomial in $n$
when $k>1$~\footnote{
For $k=1$, the Levenstein code~\cite{levenshtein1966binary,VT65} 
achieves optimal $\log n+O(1)$ bits.}. 
Our solution is based on a modification of the randomized one 
described in~\cite{IMS05}, in which we replace randomized string 
signatures (like Rabin-Karp hash function~\cite{karp1987efficient}) 
with deterministic ones~\cite{Vi90}. As an immediate application of 
our new protocol, we show a new error correcting code that is efficient 
for large numbers of (adversarial) edit errors. 
This improves on the code recently 
shown in~\cite{chakraborty2015low}, which works only for a very 
small number of errors. 

\section{Tools and Prelminaries}
In this section, we describe the main tools and techniques that 
will be used in our solution.  
\subsection{Strings, periods and deterministic samples}
Our main tools will be from string algorithmic literature. 
Recall that a string $p[1..m]$ is a sequence 
of $m$ characters from alphabet $\Sigma$. In this paper, 
we are mostly interested in $\Sigma=\{0,1\}$. We denote 
by $pq$ or $p\cdot q$ the string that consists in 
the concatenation of string $p$ with string $q$. 
We denote by $p[i..j]$, the substring of $p$ that 
spans positions $i$ to $j$. We denote by $|p|$ the length of string $p$. 
The edit-distance between two strings $p$ and $q$ is defined 
as the minimal number of edit operations necessary to transform 
$p$ into $q$, where the considered operations are character insertion, 
deletion, or substitution. 
We denote by $p^c$ the string that consists in the 
concatenation of $c$ copies of string $p$.

We now give some definitions about string periodicities. 
Recall that a string $p$ has period $\pi$
if and only if $p$ is prefix of $(p[1..\pi])^k$ for some integer 
constant $k>0$. 
An equivalent definition states that $\pi$ is a period 
of $p$ if $p[i]=p[i-\pi]$ for all $i\in[\pi+1,m]$. 
If $\pi$ is the shortest period of $p$, then 
$\pi$ is simply called \emph{the period} (we will usually mention when 
we talk about an arbitrary period that is not necessarily the shortest). 
We will make use of some easy simple properties of periods: 
\begin{enumerate}
\item Let $\pi$ a period of a string $p$. Then $\pi$ will also be a period 
of any substring of $p$ of length at least $\pi$. 
\item Given two strings 
$p$ and $q$ of same length and same period $\pi$, then $p=q$ if and only 
if $p[1..\pi]=q[1..\pi]$. 
\end{enumerate}
We will also use this lemma whose (trivial) proof is omitted. 
\begin{lemma}
\label{lemma:period_merge}
Suppose that a string $T$ has two substrings $T[i..j]$ 
and $T[i'..j']$ such that: 
\begin{enumerate}
\item $i'>i$ and $j'>j$ (none of the two substrings is included in the other). 
\item $\pi$ is a period of both strings. 
\item $j\geq i'+\pi$ (the overlap between the two substrings is at least $\pi$). 
\end{enumerate}
Then $\pi$ will also be period of substring $T[i..j']$. 
\end{lemma}


We will also make extensive use of the periodicity lemma due 
to Fine and Wilf~\cite{fine1965uniqueness}: 

\begin{lemma}
Suppose that a string $p$ of length $m$ has two periods 
$\pi_1$ and $\pi_2$ such that $p+q-\mathtt{gcd}(p,q)\leq m$, 
then it will also have period $\pi_3=\mathtt{gcd}(p,q)$. 
\end{lemma}
However, we will apply it only when we have 
two periods of lengths at most $m/2$: 
\begin{lemma}
\label{lemma:period_lemma0}
Suppose that a string $p$ of length $m$ has two periods 
$\pi_1,\pi_2\leq m/2$, 
then it will also have period $\pi_3=\mathtt{gcd}(p,q)$. 
\end{lemma}

The following lemma is easy to prove
using Lemma~\ref{lemma:period_lemma0}: 
\begin{lemma}
\label{lemma:period_lemma}
Given a string $p$ of length $m$ with period $\pi$, then for any $k\leq \mathtt{min}(\pi-1,m/3)$, 
we can always find a substring of $p$ of length $3k$ whose period is 
more than $k$. 
\end{lemma}
\begin{proof}
The proof is by contradiction. Suppose that every substring 
of length $3k$ has period at most $k$. 
Let the period of 
$p[1..3k]$ be $\pi_1\leq k$. Then suppose the period of $p[\pi_1+1..\pi_1+3k]$ 
is $\pi_2\leq k$. Then $\pi_2=\pi_1$ by basic periodicity lemma, since 
if it was not the case, then $\pi'=\mathtt{gcd}(\pi_1,\pi_2)\leq \pi_2/2$ will be period 
of $p[\pi_1+1..\pi_1+\pi_2]$ (and hence period of $p[\pi_1+1..\pi_1+3k]$) and it 
will also be period of $p[1..\pi_1]$ (and hence of $p[1..3k]$). Hence, we have a contradiction 
and $p[1..3k]$ and $p[\pi_1+1..\pi_1+3k]$ will have same period $\pi_1$. 
This implies that $\pi_1$ is period of $p[1..\pi_1+3k]$. 
We continue in the same way until we reach string $p[1..i\pi_1+3k]$ with $i=\lfloor (m-3k)/\pi_1\rfloor$.
We thus have that $p[j-\pi_1]=p[j]$ for all $j\in[\pi_1+1..,i\pi_1+3k]$. Before continuing, we notice 
that $m-3k\leq (i+1)\pi_1-1$ and thus $i\pi_1+3k\geq m-\pi_1+1\geq m-k+1$. We can now state that 
$p[j-\pi_1]=p[j]$ for all $j\in[\pi_1+1..,m-k+1]$. 

At this point, we let $\pi_0=\pi_1$ and consider the string $p'=p[m-3k+1,m]$. Assume this string has a period 
$\pi'\neq \pi$ with $\pi'\leq k$. Then the substring $p''=p[m-3k+1,i\pi_1+3k]$ has length at least $2k$, since $i\pi_1+3k\geq m-\pi_1\geq m-k$. 
Now by periodicity lemma $p''$ will have periods $\pi_0\leq k$ and $\pi'\leq k$ and thus will also have period 
$\pi''=\mathtt{gcd}(\pi_0,\pi')$. Now this implies that $\pi''$ is also shortest period of $p'$ since prefix 
$q'$ of $p'$ of length $\pi'$ is of period $\pi''$ and thus string $p'$ is prefix of $(p'[1..\pi'])^{c_1}=((q')^{c_2})^{c_1}=(q')^{c_2c_1}$ for 
some constants $c_1,c_2\geq 2$. Thus we have a contradiction with the fact that $\pi'$ is shortest period 
of $p'$. We thus have proved that $\pi_0$ is period of $p'$ and thus that $p'[j]=p'[j-\pi_0]$ for all $j\in[\pi_0+1,3k]$. 
This is equivalent to the fact that $p[j]=p[j-\pi_0]$ for all $j\in[m-3k+\pi_0+2,m]$
and thus all $j\in[m-2k+2,m]$. Since we already had $p[j]=p[j-\pi_0]$ for all $j\in[\pi_0+1..,m-k+1]$, 
we conclude that $p[j]=p[j-\pi_0]$ for all $j\in[\pi_0+1..,m]$ and thus $\pi_0\leq k$ is period of $p$,
a contradiction. 
\end{proof}
We will also use the following lemma by Vishkin~\cite{Vi90}
about deterministic string sampling. 
\begin{lemma}~\cite{Vi90}
\label{lemma:det_string_sampling}
Given a non-periodic pattern $p$ of length $m$ 
we can always find a set $S\subset [1..m]$ with $|S|\leq \log m-1$
and a constant $k<m/2$ 
such that given any text $T$, if $T[i..i+m-1]$ matches $p$
at all positions in $S$ then $T[j..j+m-1]\neq p$ for all $j\in [i-k..i-1]\cup[i+1..i-k+m/2]$.
Moreover, the set $S$ and the constant $k$ can be determined in time $O(m)$. 
\end{lemma}

We can then immediately prove the following lemma: 
\begin{lemma}
\label{lemma:period_det_string_sampling}
Given a pattern $p$ of length $m$ 
and period $\pi\leq m/3$
we can always find a set $S\subset [1..\pi]$ with $|S|\leq \log\pi$
such that given any text $T$, if $T[i..i+m-1]$ has period 
(not necessarily shortest) $\pi$ and matches $p$ at all positions in $S$ then $T[j..j+m-1]\neq p$ 
for all $j\in [i+1..i+\pi-1]$. Moreover, the set $S$ can be determined 
in time $O(m)$. 
\end{lemma}
\begin{proof}
The lemma can be proved by using Lemma~\ref{lemma:det_string_sampling}
twice. Let $p'=p[1..2\pi-1]$. It is easy to see that $p'$ is non-periodic.
If it was, then it would have another period $\pi'<\pi$ and 
$p$ would have period $\mathrm{gcd}(\pi,\pi')<\pi$, a contradiction. 
Now, by Lemma~\ref{lemma:det_string_sampling} applied on $p'$, we can find a set 
$S\subset [1..2\pi-1]$ with $|S|\leq \log\pi+1$ and a constant $k<\pi$ 
such that given any text $T$, if $T[i..i+2\pi-1]$ matches $p'$ at positions 
$S$, then $T[j..j+2\pi-1]\neq p'$ for all $j\in   [i-k..i-1]\cup [i+1..i-k+\pi-1]$. 
Now assume that $T[i..i+m]$ is periodic with period $\pi\leq m/3$ 
then the fact that $T[i..2\pi-2]$ matches $p'$ at positions $S$, implies that 
$T[j..j+2\pi-1]\neq p'$ for all $j\in [i+1..i-k+\pi-1]$. 
Since $T[i..3\pi-1]$ is periodic too, then $T[i..2\pi-2]=T[i+\pi..3\pi-2]$
and $T[i+\pi..3\pi-2]$ matches $p'$ at positions $S$ and applying the 
lemma again we get that $T[j..j+2\pi-1]\neq p'$ for all $j\in[i+\pi-k..i+\pi-1]$. 
Thus we have have that  $T[j..j+2\pi-1]\neq p'$ for all $j\in[i+1..i+\pi-1]$. 
Also, since $p'$ and $T[i..2\pi-1]$ have both period $\pi$, comparing any 
positions in $S$ reduces to comparing characters at positions in a set 
$S'\subset [1..\pi]$ with $|S'|\leq |S|$. This finishes the proof of the lemma. 

\end{proof}

\subsection{Supporting algorithms}
In this subsection, we describe some supporting algorithms. 
These are only useful for efficient implementation 
(time) of our main algorithms. Our main results can be understood without 
the need to understand these algorithms. They are usually not the most efficient 
ones, but we tried to find the simplest algorithms that run in optimal time,  
up to logarithmic factors. All lemmas are folklore or easily follow 
from known results. 

\begin{lemma}
\label{lemma:runs}
Given a string $T$ of length $n$ and fixed length $m$ we can compute the periods
of all periodic substrings of $T$ of length $m$ in time $O(n\log n)$. 
\end{lemma}
\begin{proof}
We use the algorithm~\cite{kolpakov1999finding} to compute all the runs 
of string $T$. This will allow to compute the run of any periodic 
substring. In~\cite{kolpakov1999finding} it is proved that the number 
or runs in a string of length $n$ is $O(n)$ and moreover 
the set of all runs can be computed in $O(n)$ time. Given a text $T$, 
a run is a substring $s=T[i..j]$ so that $s$ is periodic with period 
$\pi\leq (j-i+1)/2$ and $\pi$ is not a period of $T[i-1..j]$ and $T[i..j+1]$. 
In other words, runs are the maximally long periodic substrings 
of $T$ and any periodic substring or $T$ of period $\pi$ will be substring 
of a run with the same period. Given a length $m$, we can use the runs to determine the periods 
of substrings of length $m$ of $T$ in time $O(n\log n)$. We first make 
the observation that the period of substring of $T$ is the shortest 
among the periods of all runs that include the substring. 
We can thus show the following algorithm. 
We put the runs into two lists, a list $L_s$ sorted by increasing starting positions 
and another list $L_e$ sorted by increasing ending positions. We then for $i$ increasing 
from $1$ to $n-m+1$ do the following: 
\begin{enumerate}
\item Check if the next run in the list $L_s$ has starting position 
$i$ and if so insert into the binary search tree and advance the list pointer. 
\item Scan the list $L_e$ until reaching a run that ends in position less than $i$ and remove all the encountered 
runs from the binary search tree and update the list pointer to the successor of the last 
removed run. 
\item Finally, set the period of string $T[i..i+m-1]$ to be the smallest of among 
the periods of runs currently stored in the binary search tree. 
\end{enumerate}
The correctness of the algorithm stems from the fact that at any 
step $i$ the binary search tree will contain exactly the runs that 
span substring $T[i..i+m-1]$. 
Concerning the running time, it is clear that every operation takes at most $O(\log n)$ time
and thus running time of all steps is $O(n\log n)$. In addition, the slowest 
operation in the preprocessing is the sorting of the lists $L_s$ and $L_e$ 
which takes $O(n\log n)$. Thus, the whole algorithm runs in time $O(n\log n)$. 
This finishes the proof of the lemma.
\end{proof}

\begin{lemma}~\cite{knuth1977fast}
\label{lemma:KMP}
We can determine the period of a string of length 
$n$ in time $O(n)$. 
\end{lemma}

\begin{lemma}
\label{lemma:substring_period_query}
Given a string $T$ of length $n$, we can build a data structure 
of size $O(n)$, so that we can check whether the period of any substring 
$T[i..j]$ has period $\pi$ (not necessarily shortest one) in time $O(1)$. 
\end{lemma}
\begin{proof}
We build suffix tree on $T$~\cite{mccreight1976space} with support for Lowest common ancestor 
queries~\cite{bender2000lca}. This will allow to answer longest common prefix queries 
between substrings of $T$. Then checking whether 
\end{proof}
\subsection{Error correcting codes}
We will make use of the systematic error correcting codes. 
Given a length $n$ and a parameter $k<n/2$, one would wish to have an algorithm 
that takes any string $s$ of length $n$ bits and encodes 
it into a string $s'$ of length $f(n,k)$ such that 
one can recover $s$ from $s'$ even if up to $k$ positions 
of $s$ are corrupted. Reed-Solomon codes~\cite{reed1960polynomial} 
are a family of codes in which $f(n)=n+\Theta(k\log n)$. A code is said to be 
systematic if $s'$ can be written as concatenation 
of $s$ with a string $r$ of $\Theta(k\log n)$ bits. The string $r$ is 
called the redundancy of the code. 
In fact a Reed-Solomon code 
can be used to correct a string of length $n$ over alphabet 
$[1..\Theta(n)]$ against $k$ errors using the same redundancy $\Theta(k\log n)$ 
bits. 

\begin{lemma}~\cite{gao2003new}
\label{lemma:ECC}
There exists a systematic Reed-Solomon code for strings 
of length $n$ over alphabet $[1..\Theta(n)]$ which can be 
encoded and decoded in time $\Theta(n\cdot\mathtt{polylog}(n))$. 
\end{lemma}

We notice that a systematic Reed-Solomon code 
can be used to implement an efficient document exchange under 
the hamming distance. Given a string $s$ of length $n$, simply 
compute the systematic error correcting code on $s$
and send the redundancy $r$. The receiver can then concatenate 
his own string with $r$ and use the decoding algorithm. 
Clearly if the receiver's string differs
from $s$ in at most $k$ positions, then the decoding 
algorithm will be able to recover the string $s$.

\section{Document Exchange for Edit Distance}

Our scheme is based on the one devised by 
Irmak, Mihaylov and Suel~\cite{IMS05} (henceforth denoted IMS). The scheme 
is randomized (Monte-Carlo). Our contribution is to show 
how to make the scheme deterministic at the cost of a slight 
increase in message size. 
\subsection{IMS Randomized protocol}
The IMS scheme works as follows: given a string $T_A$ of length 
$n$ without loss of generality assume that $n=2^bk$ 
for some integer $b$ and that $k$ is a power of two. 
Divide the string into $2k$ pieces of equal lengths 
and send the hash  
signatures on each piece. 
Then divide the string into $4k$ pieces 
and send the signatures compute the hash 
signatures of all pieces (each string has hash signature 
of length $c\log n$ for some constant $c$), compute systematic error 
encoding code (Reed-Solomon for example) 
on them with redundancy $2k$ and send 
the redundancy. We do the same at all levels until we reach 
$\Theta(n/\log n)$ pieces of length $c\log n$ bits, 
for some constant $c$ in which case, we send redundancy 
of length $\Theta(k\log n)$ bits on the piece's content. At the end, 
we will get $\log (n/k)$ levels, at each level
sending $\Theta(k\log n)$ bits for total of $\Theta(k\log(n/k)\log n)$ bits. 
The receiver holding a string $T_B$ of length $n$ with edit distance 
at most $k$ from string $T_A$ can recover $T_A$ solely from the message
as follows. At root level he tries to match every signature $i$ with 
substrings of $T_B$ of length $B$ starting at positions $[iB+1-k,iB+1+k]$, 
where $B=n/(2k)$ at each step comparing the hash signatures. 
If any hash signature matches we conclude that strings are equal and we copy the block.
The main idea is that we can match all but $k$ pieces~\footnote{This comes 
from a basic property of edit distance which states that if string 
$T_B$ is at edit distance $k$ from string $T_A$, then at all but $k$
blocks of $T_B$ can be found in $T_A$, and moreover their positions 
in $T_B$ is shifted from their position in $T_B$, by at most $k$
positions.}. 
For each piece that matches
a substring of $T_B$, we can divide the substring into two pieces and compute the hash 
signatures on it. For pieces that are not matched we divide them into two pieces 
and associate random signatures with them. Thus at next step, we can build $4k$ signatures and be 
assured that at most $2k$ hash signatures could be wrong. We then correct the $2k$ wrong hash signatures 
using the redundancy and continue. At next step for every $i\in[1..4k]$, try to match 
signature number $i$ with substrings of $T_B$ of length $B$ that start at positions $[iB+1-k,iB+1+k]$, 
this time with $B=n/(4k)$. We then induce (up to) $8k$ signatures from the matching substrings of $T_B$ 
(if a signature of a block of $T_A$ does not match any string of $T_B$, we put $2$ arbitrary signatures), 
and be assured that at most $2k$ signatures are wrong. We continue in the same way, at each step deducing the hash signatures 
at consecutive levels, until we reach the bottom level, at which we copy the content of 
each matching blocks from $T_B$ instead of writing two signatures. We then can deduce the 
content of $T_A$ after correcting the $n/B$ copied blocks each of length $B=c\log n$ bits
(at most $2k$ blocks are wrong). 
The whole algorithm works with high probability, for sufficiently large constant $c$, since 
a substring of $T_B$ will match a signature of a substring of $T_A$ if they are equal 
and will not match with high probability if they differ. 
\subsection{Our Deterministic protocol}

We will show how to use the deterministic signatures of 
Lemma~\ref{lemma:det_string_sampling} to 
make the IMS protocol deterministic. Before giving formal 
details, we first give an overview of the modification. 
First, recall that at each level, we need to find for 
each block of $S_A$, a matching substring from $S_B$. 
This matching string will have to be in a window 
of size $2k+1$. In the randomized scheme, the
signatures will allow to ensure that with high probability a 
substring will match a block only if 
it is equal. Our crucial observation is that we are allowed
to return an arbitrary false positive match if \emph{no}
substring in the window matches the block, since the error-correcting
code will allow to recover the information. However, in case of a 
match, we will have to return only the matching substring
(no false negatives or false positives are allowed). 
By using deterministic samples and exploiting properties 
of string periodicities we will be able to eliminate 
all but one matching substring as long as the compared strings 
are long enough (the length has to be at least $ck$ 
for some suitable constant $c$). Due to this constraint 
our scheme will not work at bottom levels. We thus stop 
using it at the first level with blocks of size $ck$ 
bits and instead store the redundancy to allow 
to recover these blocks, incurring $\Theta(k^2)$ more 
bits of redundancy. 
We now give more details on our scheme. 
\paragraph{Encoding}
We reuse the same scheme as above (the IMS scheme) but this time using deterministic signatures
and stopping at level with pieces of length $\mathtt{max}(32k,2^{\lceil\log\log n\rceil})$ bits. 
At each level (except the bottom), the signature of a piece $p=T_A[iB+1,iB+B]$
will consist in the following information: 
\begin{enumerate}
\item Let $\pi$ be the period of $p$. We will store the starting 
position $s$ and length $\ell$ of some substring $p'$ of $p$. 
If $\pi\leq 4k+2$ then we set $p'=p$, $s=0$ and $\ell=B$. 
Otherwise, Let $p''$ be a substring of $p$ of length $12k+6$ and period $\pi''$ longer than $4k+2$ (this 
is always possible by Lemma~\ref{lemma:period_lemma}, since the string 
$p$ has period more than $4k+2$). If $p''$ is non-periodic, then 
we let $p'=p''$, otherwise set $p'=p''[1..2\pi-1]$. Notice that 
$|p'|\geq 8k+4$ and $p'$ is always non-periodic. 
The starting position and length need $O(\log n)$ bits. 
\item A deterministic sample of $p'$
which stores $\Theta(\log k)$ positions of $p'$
and the value of the characters at those positions. 
Each position is stored using $\Theta(\log k)$ bits, since 
in case $p'=p''$, we have $|p''|\leq 12k+6$ and in 
case $p'=p$, all samples are from first $\pi\leq 4k+2$
positions. We also store the period $\pi'$ of $p'$. In total 
we store $O(\log n+\log^2k)$ bits. 
\end{enumerate}
It is clear that the total information stored for each piece will be 
of length $O(\log n+\log^2k)$ bits. 
At the bottom, level, the 
string $T_A$ is divided into pieces of length $n/B$ 
with $B=\mathtt{max}(32k,2^{\lceil\log\log n\rceil})$ and the stored redundancy will be $\Theta(2k(B+\log n))=\Theta(k(k+\log n))$
bits allowing to recover $2k$ pieces each of length $\mathtt{max}(32k,2^{\lceil\log\log n\rceil})$ bits. At the other levels the 
redundancy will allow to recover up to $2k$ wrong piece 
signatures, necessitating redundancy $O(k(\log^2k+\log n))$ bits. 
It remains to describe more precisely how the redundancy 
is generated. At each level, we will have a sequence of $n/B$ signatures 
of length $r_B$. Since a Reed-Solomon code for a string of length $n$ deals only with alphabet 
size up to $n$, we will divide each signature into $d=\Theta(r_B/\log (n/B))$ blocks each of length $\Theta(\log (n/B))$
bits, build $d$ sequences where a sequence $i\in[1..d]$ consists in the concatenation 
of block $i$ from all successive signatures. It is clear that the redundancy is as stated 
above and that the encoding will allow to recover from up to $2k$ wrong signatures. 
\paragraph{Decoding}
The recovery will be done now at each level by matching
every piece's signature against $2k+1$ consecutive substrings 
of $T_B$. The main idea is that at most one substring 
could match (multiple substrings could match only if they 
are equal). More in detail for matching signature of string 
$T_A[iB+1,iB+B]$ against substrings starting at positions 
$j\in [iB-k,iB+k]$ in $T_B$ we match the substring $q'=T_B[j+s,j+s+\ell-1]$ 
against the signature of string $p'$. We first determine whether  
$\pi'$ is period (not necessarily shortest) of $q'$ and if so, compare the signatures 
of $q'$ and $p'$ (comparing the substrings at sampled positions). 
We then can keep only one position as follows: 
\begin{enumerate}
\item If $\ell<B$, we know that $\pi>4k+2$, and by Lemma~\ref{lemma:det_string_sampling}, 
we can eliminate all but one candidate position $j$. 
To see why, notice that $|p'|\geq 8k+4$ and thus by ~\ref{lemma:period_det_string_sampling}
we can eliminate $t$ candidates to the left and $t'=4k+2-t-1$ candidates to the right
for some $t\geq 0$. 
Notice that either $t$ or $t'$ has to be at least $2k$. It is clear, then that 
we can not have two candidates $j$ and $j'$ within distance $2k$ without 
one of the two eliminating the other. 
\item If $\ell=B$, we have $\pi\leq 4k+2$, and the checking will work correctly
since we have $3\pi=12k+6\leq 32k$
and so Lemma~\ref{lemma:period_det_string_sampling} 
applies and any matching location $j$ will allow 
to eliminate the next $\pi-1$ positions. Moreover, all following substrings 
of period $\pi$ starting at locations at least $j+\pi$ will have 
to be equal to some substring starting at location in $[j,j+\pi-1]$. 
This is easy to see. Let $j'\geq j+\pi$ be such a position and let $j''=((j'-j)\bmod \pi)+j$. 
Let $q'=T_B[j'..j'+B-1]$ and $q''=T_B[j''..j''+B-1]$. 
Since $j'\in[j..j+B-\pi+1]$, we can apply Lemma~\ref{lemma:period_merge}, 
and induce that $Q=T_B[j..j'+B-1]$ has also period $\pi$. Since $q''=T_B[j''..j''+B-1]$ is substring
of $Q$, we deduce that it has period $\pi$ as well. Also  
$q''[1..\pi]=T_B[j''..j''+\pi-1]=T_B[j'..j'+\pi-1]=q'[1..\pi]$, 
since $Q$ has period $\pi$. Thus $q'=q''$, since they both have period $\pi$
and their first $\pi$ characters are the same. 
Thus we can eliminate all positions except ones that are actually 
the same strings (they are all $\pi$ positions apart). 
\end{enumerate}
\paragraph{Runtime analysis}
It remains to show that the protocol can be implemented in both sides with running time $O(n\cdot\mathtt{polylog}(n))$.
We start with the sender. At each level, the sender needs to compute the the period 
of each piece which can be done in time $O(B)$ using Lemma~\ref{lemma:KMP}. If the period 
is longer than $4k+2$, then it needs to find a substring of the piece of length $12k+6$ with period more than $4k+2$. 
This can be done by computing the periods of all periodic substrings in time $O(B\log B)$ using Lemma~\ref{lemma:runs}. 
Then at least one of the strings 
should be non-periodic or should be periodic with period more than $4k+2$. If it was periodic we already have its period, 
otherwise, we compute its period using Lemma~\ref{lemma:KMP}. 
Then the deterministic sample for substring $p'$ associated 
with a piece is also computed in $O(B)$ time according 
to lemmas~\ref{lemma:det_string_sampling} and~\ref{lemma:period_det_string_sampling}. 
Summing up over all $n/B$ pieces we get 
running time $O((n/B)B\log B)\in O(n\log n)$ for computing the signatures of all pieces. 
Then computing the redundancy of the Reed-Solomon encoding of the concatenation 
of the pieces' encoding can be done in time $O((n/B)\cdot\mathtt{polylog}(n/B)\frac{\log^2k+\log n}{\log n})$
according to Lemma~\ref{lemma:ECC}.
Thus the total computation time at each level is
$O(n\cdot\mathtt{polylog}(n))$ and at all $\log (n/k)$ levels is also $O(n\cdot\mathtt{polylog}(n))$. 
This finishes the analysis of computation time at the sender's side. 
It remains to show that the time spent on the receiver's side can also be upper bounded 
by $O(n\cdot\mathtt{polylog}(n))$. The only non-trivial steps are:
\begin{enumerate}
\item The decoding of Reed-Solomon encoded strings, which can be done within the 
same time as encoding. 
\item The determination of whether a given substring 
from the receiver side has a certain period $\pi$ which 
can be determined in constant time after preprocessing 
of the whole string in $O(n)$ time (see Lemma~\ref{lemma:substring_period_query}). 
\end{enumerate}
All other steps are either trivial or identical to the ones 
on the sender's side. 
Thus we have that the running time on both sides is $O(n\cdot\mathtt{polylog}(n))$. 

We thus have proved the main theorem of this paper: 
\begin{theorem}
\label{main_theo}
There exists a one-way deterministic protocol 
for document exchange under the edit distance with  
running time $O(n\cdot\mathtt{polylog}(n))$
and message size $O(k^2+k\log^2n)$ bits. 
\end{theorem}


\section{A Scalable Error Correcting Code}

In~\cite{brakensiek2015efficient} it was shown that 
one can construct an efficient error correcting code 
for (adversarial) edit errors~\footnote{In this model, the receiver must 
be able to recover the original string regardless of the locations 
of errors, as long as their number is bounded by some parameter $k$~\cite{schulman1999asymptotically}.
} with near-linear encoding-decoding time and redundancy 
$O(k^2\log k\log n)$ bits. However, the analysis of both the time and 
redundancy assumes that $k$ is very small compared to $n$. 
In fact the scheme is not even defined for values of $k$ as 
small as $\sqrt{\log\log n}$. 


We can use our main result to construct an error correcting code 
with redundancy $r=O(k^3+k^2\log^2 n)$ bits 
and encoding-decoding time complexity $O(n\cdot\mathtt{polylog}(n))$. 
The scheme works for $k$ as large as $O(n^{1/3})$. 
Given the input string $s$ of length $n$, we first construct the message 
described in Theorem~\ref{main_theo}. The size of this message 
is $O(k^2+k\log^2 n)$ bits. 
Then we protect this message against $k$ edit operations 
by using a $(2k+1)$-repetition code
similarly to~\cite{brakensiek2015efficient}\footnote{They actually use a $(3k+1)$-repetition code, 
but as our analysis shows below, a $(2k+1)$-repetition code is sufficient in our case.}. 
The size of the protected 
message is then $m=r(2k+1)$ (each input bit is duplicated $2k+1$ times in the 
output). We finally send the original string followed by the protected 
message for a total size $n+m$ bits. 
Then, the receiver can consider 
the first $n$ bits as the original (potentially corrupted) message to be corrected, and then 
consider the remaining bits (between $m-k$ and $m+k$), as the protected message. 
The original message can then be recovered from the protected message as follows. 
Divide the protected message into $r$ blocks of length exactly $2k+1$ (except 
the last one which can have any length in $[1,3k+1$). Then, for every block, output 
the majority bit in that block. One can easily see that one edit error 
in the protected message will change the count of ones in the block it occurs in and following 
blocks by at most $\pm 1$. Thus, the decoding of the protected message will be robust 
against $k$ edit errors. 
Also, it is easy to see that $k$ edit operations on the $n+m$ sent bits 
will imply that the first $n$ bits will differ from the original string by at most 
$k$ operations and the last $m-k$ to $m+k$ bits will differ from the original 
protected message by at most $k$ edit operations. After decoding the protected 
message and recovering the original message, the receiver uses it in combination 
with the (potentially corrupted) received string to recover the original string. 
Finally, analyzing the running time, 
the only additional step we do compared to the document exchange protocol is 
the encoding and decoding of the message which takes time $O(m)\in O(n)$. 
This shows that the total encoding and decoding time for our proposed code is  
$O(n\cdot \mathtt{polylog}(n))$. We thus have shown the following theorem: 

\begin{theorem}
There exists an error correcting code 
for (adversarial) edit errors with
redundancy  $O(k^3+k^2\log^2n)$ bits and 
encoding-decoding time $O(n\cdot \mathtt{polylog}(n))$. 
The scheme works for any $k\in O(n^{1/3})$. 
\end{theorem}

\section*{Acknowledgements}
This work was initiated when the author was visiting MADALGO 
(Aarhus University) in Spring 2012. The author wishes to thank Hossein Jowhari 
and Qin Zhang for the initial fruitful discussions that motivated 
the problem. 

\bibliographystyle{plain}
\bibliography{document_exchange.bib}

\end{document}